\documentclass[a4paper,12pt]{article}
\usepackage{amssymb,amsmath,bm}
\usepackage{graphicx}
\usepackage{enumerate}
\usepackage{amsthm}

\newtheorem{theorem}{Theorem}
\newtheorem{step}{Step}
\newtheorem{fact}{Fact}

\makeatletter
\renewenvironment{proof}[1][\proofname]{\par
  \pushQED{\qed}%
  \normalfont \topsep6\p@\@plus6\p@\relax
  \trivlist
  \item\relax
        {\bf
    #1\@addpunct{.}}\hspace\labelsep\ignorespaces
}{%
  \popQED\endtrivlist\@endpefalse
}
\makeatother

\title{A Rigorous Proof of the Index Theorem for Economists}
\author{Yuhki Hosoya\thanks{TEL: +81-90-5525-5142, E-mail: ukki(at)gs.econ.keio.ac.jp}\\ Faculty of Economics, Chuo University\thanks{742-1, Higashinakano, Hachioji-shi, Tokyo, 192-0393, Japan.}}
\date{}
\begin{document}
\maketitle

\begin{abstract}
This paper provides a rigorous and gap-free proof of the index theorem used in the theory of regular economy. In the index theorem that is the subject of this paper, the assumptions for the excess demand function are only several usual assumptions and continuous differentiability around any equilibrium price, and thus it has a form that is applicable to many economies. However, the textbooks on this theme contain only abbreviated proofs, and no known monograph contains a rigorous proof of this theorem. Hence, the purpose of this paper is to make this theorem available to more economists by constructing a readable proof.

\vspace{12pt}

\vspace{12pt}

\vspace{12pt}
\end{abstract}

\section{Introduction}
How many equilibrium prices exist? Arrow and Debreu (1954) showed that at least one equilibrium price exists. It seems that economists of this period believed that the fewer the number of equilibria, the better. Therefore, Arrow et al. (1959) proposed a sufficient condition for the existence of a unique equilibrium price. However, their condition is very complicated and does not seem to be valid in many economies.

Debreu (1970) proved that the number of equilibrium prices is at least finite in ``almost all'' economies. In his view, it was difficult to prove that there is exactly one equilibrium price, and thus he proved that at least there are only a finite number of equilibria. This result had been refined and reformulated by many economists. This research area is now known as the theory of regular economy.

The most central result in the theory of regular economy is called the index theorem. This theorem states that the sum of ``indices'' computed by a certain method for equilibrium prices is equal to $1$. This result is powerful. Because ``indices'' must be able to be summed, the number of equilibrium prices must be a finite odd number. We can immediately reproduce the result of Debreu (1970) using this theorem.

The problem is that no straightforward proof of this index theorem is known. The most famous index theorem currently published in a form applicable to many economies appears to be Proposition 5.6.1 of Mas-Colell (1985). However, the proof of this theorem contains many gaps that ordinary readers will find impossible to fill.

Therefore, in this paper, we focus on the proof of this result and provide a rigorous and gap-free proof of this theorem in the shortest manner. The proof presented in this paper requires basic knowledge of differential topology supplied by Guillemin and Pollack (1974). Therefore, in this paper, we construct a proof under the assumption that the contents of the above book are known. However, since Guillemin-Pollack's formulation of differential topology is a bit peculiar, we have chosen to extract only the claims necessary for this paper and offer them to readers.

The proof in this paper requires a classification theorem for one-dimensional manifolds. However, the problem is that a well-written proof of the classification theorem contained in a book is missing. Therefore, in the appendix of this paper, we present a proof of this classification theorem.

The structure of this paper is as follows. In Section 2, we explain a part of the knowledge of differential topology contained in Guillemin-Pollack's book that is necessary to understand the proof. We do not prove any of the results that appear in this section, but the proofs are either in Guillemin-Pollack or can be easily reconstructed from what is written in Guillemin-Pollack. Section 3 contains the assertion of the main result and its proof. In Section 4, we explain the relationships between Proposition 5.6.1 of Mas-Colell (1985) and our main result, and the proof and its gaps in Mas-Colell's book. In the appendix, we prove the classification theorem for one-dimensional manifolds.

\section{Preliminaries}
In this paper, we assume that readers are sufficiently familiar with the results in Guillemin and Pollack (1974). In particular, we assume that readers know the definition of manifolds, tangent spaces and the tangent bundle, the derivatives of smooth functions on manifolds, transversality, and oriented manifolds. However, the definition of manifolds in Guillemin and Pollack (1974) is not standard in differential topology, and with very few exceptions such as Milnor (1965), another definition has been used. Thus, readers may not be familiar with the manifolds treated in this paper, even though they have appropriate deep knowledge on manifolds. Therefore, in this section, we present the definitions of the terms in differential topology treated in this paper and the assertion of some theorems. All of the results in this section have been proved or can be easily verified from the results proved in Guillemin and Pollack (1974).

Note that, although Guillemin and Pollack (1974) treated only $C^{\infty}$ manifolds and functions, in most cases, their proofs work well for $C^k$ manifolds and functions, even if $k\neq \infty$. However, several techniques depend on the assumption $k=+\infty$. For example, Sard's theorem does not necessarily hold if the function is not $C^{\infty}$. The concept of the Morse function cannot be appropriately defined in $C^1$ manifolds. The tangent bundle $T(X)$ is not necessarily a $C^1$ manifold when $X$ is a $C^1$ manifold. We avoid these technical problems in this paper. First, we use Sard's theorem only for a $C^{\infty}$ function defined on a $C^{\infty}$ manifold. Second, we consider the tangent bundle only for a $C^{\infty}$ manifold. Third, we avoid the use of the Morse functions.\footnote{For the benefit of readers, we will note which part in the following subsections corresponds to which section of Guillemin-Pollack. First, all the basic knowledge of manifolds and tangent spaces is in Sections 1.1-2 of Guillemin-Pollack, except for the inverse function theorem, which is in Section 1.3. Knowledge of transversality is argued in Section 1.5, but the prerequisite is knowing the inverse image theorem, which is in Section 1.4. Knowledge of tangent bundles is included in Section 1.8. Knowledge of the manifolds with boundary is discussed in Section 2.1. The orientations are analyzed in detail in Section 3.2. The classification theorem for 1-dimensional manifolds is in appendix 2, but as noted in the main text, Guillemin-Pollack's proof cannot be adopted for our result. Hence, See our appendix.}

\subsection{Manifolds}
First, suppose that $X\subset \mathbb{R}^N$ and let $f:X\to \mathbb{R}^M$ be a continuous function. Choose any $x\in X$ and let $k\in\mathbb{N}\cup \{\infty\}$. We say that $f$ is $C^k$ around $x$ if and only if there exists a $C^k$ function $F:U\to \mathbb{R}^M$ such that $U$ is an open neighborhood of $x$ in $\mathbb{R}^N$ and $F(y)=f(y)$ whenever $y\in U\cap X$. As usual, a function that is $C^k$ everywhere is called a $C^k$ function.

A function $f:X\to Y$ is called a {\bf homeomorphism} if and only if $f$ is bijective and both $f$ and $f^{-1}$ are continuous. Suppose that $f:X\to Y$ is a homeomorphism. We say that $f$ is a $C^k$ {\bf diffeomorphism} if and only if both $f$ and $f^{-1}$ are $C^k$. A nonempty set $X\subset \mathbb{R}^N$ is called an {\bf $n$-dimensional $C^k$ manifold} if and only if for every $x\in X$, there exists a $C^k$ diffeomorphism $\phi:U\to V$ such that $U$ is an open set in $\mathbb{R}^n$ and $V$ is an open neighborhood of $x$ in $X$. Such a $\phi$ is called a {\bf local parametrization around $x$}. The number $n$ is called the {\bf dimension of $X$} and denoted by $\dim X$. It is known that if $X$ is a $C^k$ manifold, then $\dim X$ is uniquely determined.

Sometimes, $C^0$ manifolds are considered. A function is said to be $C^0$ if and only if it is continuous, and a set $X\subset \mathbb{R}^N$ is called an $n$-dimensional $C^0$ manifold if and only if every $x\in X$ has an open neighborhood that is homeomorphic to some open set in $\mathbb{R}^n$. A $C^0$ manifold is sometimes called a {\bf topological manifold}. However, in $C^0$ manifolds, we cannot define the tangent space, which is explained later. This is very restrictive for our argument, and thus we do not use a $C^0$ manifold in this paper, except in the appendix. In this connection, we always assume that $k\ge 1$ in the main text.

We define $\mathbb{R}^0=\{0\}$. Hence, $0$-dimensional manifolds can be defined. A set $X$ is a $0$-dimensional $C^k$ manifold if and only if it is a discrete set.

Let $X$ be an $n$-dimensional $C^k$ manifold, and choose any $x\in X$. Choose any local parametrization $\phi$ around $x$. Let $\phi(z)=x$ and suppose that $D\phi(z)$ denotes the Fr\'echet derivative of $\phi$ at $z$. Let $T_x(X)$ denote the range of $D\phi(z)$. Then, it is known that $T_x(X)$ is independent of the choice of $\phi$. We call this linear space $T_x(X)$ the {\bf tangent space} of $X$ at $x$. It is known that the dimension of $T_x(X)$ is $n$, and for every $v\in \mathbb{R}^N$, $v\in T_x(X)$ if and only if there exist $\varepsilon>0$ and a $C^1$ curve $c:(-\varepsilon,\varepsilon)\to X$ such that $c(0)=x$ and $c'(0)=v$.

Next, let $X,Y$ be $C^k$ manifolds, and let $f:X\to Y$ be a $C^k$ function. Choose $x\in X$, and a local parametrization $\phi$ around $x$. Let $\phi(z)=x$. Then, $g=f\circ \phi$ is a $C^k$ function defined on an open set $U\subset \mathbb{R}^n$, where $n$ is the dimension of $X$. Hence, the Fr\'echet derivative $Dg(z)$ can be defined. Choose any $v\in T_x(X)$, and let $D\phi(z)w=v$. Then, we can define $df_x(v)=Dg(z)w$, and it can be shown that the left-hand side is independent of the choice of $\phi$. Moreover, $df_x(v)\in T_{f(x)}(Y)$ for all $x\in X$ and $v\in T_x(X)$. This linear function $df_x$ is called the {\bf derivative} of $f$ at $x$. Suppose that $v\in T_x(X)$ and $c:(-\varepsilon,\varepsilon)\to X$ is a $C^1$ function that satisfies $c(0)=x$ and $c'(0)=v$. For a $C^k$ function $f:X\to Y$, if we define $d=f\circ c$, then we can check that $df_x(v)=d'(0)$.

Suppose that $X,Y$ are $C^k$ manifolds whose dimensions are the same, and $f:X\to Y$ is a $C^k$ function. Choose $x\in X$ and suppose also that $df_x$ is a bijection from $T_x(X)$ onto $T_{f(x)}(Y)$. Then, there exists an open neighborhood $U$ of $x$ such that $F(U)$ is an open neighborhood of $f(x)$ and the restriction of $f$ is a $C^k$ diffeomorphism between $U$ and $f(U)$. This result is known as the {\bf inverse function theorem}.

\subsection{Submanifolds and Transversality}
Let $X$ be an $n$-dimensional $C^k$ manifold, and $Y\subset X$. We call $Y$ a {\bf submanifold} of $X$ if $Y$ itself is a $C^k$ manifold. It is known that if $Y$ is a submanifold of $X$, then $T_x(Y)\subset T_x(X)$. If $\dim Y=m$, then $n-m$ is called the {\bf codimension} of $Y$, and denoted by $\mbox{codim}~Y$.

Suppose that $X,Y$ are $C^k$ manifolds, $Z$ is a submanifold of $Y$, and $f:X\to Y$ is $C^k$. We say that $f$ is {\bf transversal to $Z$} if and only if for every $x\in f^{-1}(Z)$,
\[\mbox{Im}(df_x)+T_{f(x)}(Z)=T_{f(x)}(Y).\]
Then, the following fact is known.

\begin{fact}
Suppose that $X,Y$ are $C^k$ manifolds, $Z$ is a $C^k$ submanifold of $Y$, and $f:X\to Y$ is a $C^k$ function that is transversal to $Z$. If $S=f^{-1}(Z)$ is nonempty, then $S$ is also a submanifold of $X$, and $\mbox{codim}~S=\mbox{codim}~Z$.
\end{fact}

\subsection{Tangent Bundle}
Suppose that $X$ is an $n$-dimensional $C^{\infty}$ manifold. The following set
\[T(X)=\{(x,v)|x\in X,\ v\in T_x(X)\}\]
is called the {\bf tangent bundle} of $X$. It is known that $T(X)$ is a $2n$-dimensional $C^{\infty}$ manifold, and if $(x,v)\in T(X)$ and $\phi$ is a local parametrization around $x$ in $X$, then
\[(z,w)\mapsto (\phi(z),d\phi_z(w))\]
is a local parametrization around $(x,v)$ in $T(X)$.

By construction, it is easy to show that $T_{(x,v)}(T(X))=(T_x(X))^2$.

\subsection{Manifolds with Boundary}
We need the definition of manifolds with boundary. First, let $H^n=\{x\in \mathbb{R}^n|x_n\ge 0\}$. A set $X\subset \mathbb{R}^N$ is called an {\bf $n$-dimensional $C^k$ manifold with boundary} if and only if for every $x\in X$, there exists a $C^k$ diffeomorphism $\phi:U\to V$ such that $U$ is a relatively open set in $H^n$ and $V$ is an open neighborhood of $x$ in $X$. The mapping $\phi$ is called a local parametrization around $x$ as usual. Tangent space $T_x(X)$ can be defined in the usual manner. Note that, $T_x(X)$ is an $n$-dimensional linear space, even when $X$ is not a usual manifold but a manifold with boundary.

Suppose that $X$ is an $n$-dimensional $C^k$ manifold with boundary, and define the set $\partial X$ as the set of all $x\in X$ such that there is no local parametrization around $x$ whose domain is an open set in $\mathbb{R}^n$. It can be shown that if $\partial X$ is nonempty, then it is an $(n-1)$-dimensional $C^k$ manifold, and $T_x(\partial X)$ is a subspace of $T_x(X)$.

Note that, $X\times Y$ is not necessarily a manifold with boundary, even if both $X,Y$ are manifolds with boundary. For example, it is known that $[0,1]^2$ is not a manifold with boundary. However, if $X$ is an $n$-dimensional $C^k$ manifold with boundary and $Y$ is an $m$-dimensional $C^k$ manifold without boundary, then $X\times Y$ is an $(n+m)$-dimensional $C^k$ manifold with boundary, and $\partial (X\times Y)=\partial X\times Y$ and $T_{(x,y)}(X\times Y)=T_x(X)\times T_y(Y)$. In particular, because the unit interval $[0,1]$ is a $1$-dimensional $C^{\infty}$ manifold with boundary, if $X$ is an $n$-dimensional $C^k$ manifold, then $Y=[0,1]\times X$ is an $(n+1)$-dimensional $C^k$ manifold with boundary, and $\partial Y=\{0,1\}\times X$. If we write $X_t=\{t\}\times X$, then $\partial Y=X_0\cup X_1$ in this sense.

Suppose that $X,Y$ are $C^k$ manifolds with boundary, and $f:X\to Y$ is $C^k$. Then, we can define the derivative $df_x:T_x(X)\to T_{f(x)}(Y)$ in the same manner as in Subsection 2.1. Usually, $\partial f$ denotes the restriction of $f$ into $\partial X$. If $Y$ has no boundary and $Z$ is a submanifold of $Y$, we can define transversality as usual: that is, $f$ is transversal to $Z$ if and only if for any $x\in f^{-1}(Z)$,
\[\mbox{Im}(df_x)+T_{f(x)}(Z)=T_{f(x)}(Y).\]
Then, the following fact holds.

\begin{fact}
Suppose that $X$ is a $C^k$ manifold with boundary, $Y$ is a $C^k$ manifold without boundary, and $Z$ is a submanifold of $Y$ without boundary. Let $f:X\to Y$ be a $C^k$ function such that both $f$ and $\partial f$ are transversal to $Z$. If $S=f^{-1}(Z)$ is nonempty, then $S$ is a $C^k$ manifold with boundary whose codimension is the same as that of $Z$, and $\partial S=\partial X\cap S$.
\end{fact}

\subsection{Orientation}
Suppose that $V$ is an $n$-dimensional linear space. Choose any two bases $v=(v_1,...,v_n)$ and $w=(w_1,...,w_n)$. Then, there uniquely exists an $(n\times n)$ matrix $A$ such that
\[w_i=a_{i1}v_1+...+a_{in}v_n\]
for all $i\in \{1,...,n\}$. We say that $v$ is {\bf equivalent} to $w$ if and only if $\det A>0$.\footnote{We use both expressions $\det A$ and $|A|$ to denote the determinant of $A$.}

Suppose that $X\subset \mathbb{R}^N$ is an $n$-dimensional $C^k$ manifold with/without boundary. Let
\[T^n(X)=\{(x,v_1,...,v_n)|x\in X,\ v_i\in T_x(X)\}.\]
Choose any function $\mbox{sign}_X:T^n(X)\to \{+1,-1\}$. We say that $\mbox{sign}_X$ is {\bf sign-preserving} if for any two equivalent bases $v=(v_1,...,v_n)$ and $w=(w_1,...,w_n)$ of $T_x(X)$, $\mbox{sign}_X(x,v)=\mbox{sign}_X(x,w)$. Suppose that for every $x\in X$, there exists a local parametrization $\phi$ around $x$ such that for $v_i(z)=D\phi(z)e_i$, $\mbox{sign}_X(\phi(z),v_1(z),...,v_n(z))$ is constant, where $e_i\in \mathbb{R}^n$ is the $i$-th unit vector. Then, we call this sign-preserving function $\mbox{sign}_X$ an {\bf orientation} of $X$.

We must note an exception. Suppose that $X$ is a $0$-dimensional manifold. By definition, $X$ is a discrete set, and $T_x(X)=\{0\}$. In this case, any function $\mbox{sign}_X:X\to \{-1,+1\}$ is called an orientation of $X$.

Note that, there is a manifold with no orientation: the M\"obius strip is a typical example. If $X$ is a manifold and there is an orientation, then we say that $X$ is {\bf orientable}. If we think that an orientation of $X$ is given, then $X$ is said to be {\bf oriented}. It is known that if $X$ is a connected and orientable manifold, then there are exactly two possible orientations.

If $X,Y$ are manifolds and either $\partial X=\emptyset$ or $\partial Y=\emptyset$, then $X\times Y$ is also a manifold. If, in addition, both $X,Y$ are oriented, then $X\times Y$ is also oriented, and the orientation is made in the following manner. Suppose that $\mbox{sign}_X,\mbox{sign}_Y$ denote orientations. Then, there uniquely exists a sign-preserving function $\mbox{sign}_{X\times Y}$ such that
\begin{align*}
&~\mbox{sign}_{X\times Y}((x,y),(v_1,0),...,(v_n,0),(0,w_1),...,(0,w_m))\\
=&~\mbox{sign}_X(x,v_1,...,v_n)\mbox{sign}_Y(y,w_1,...,w_m)
\end{align*}
for all $(x,y)\in X\times Y$ and bases $v=(v_1,...,v_n)$ of $T_x(X)$ and $w=(w_1,...,w_m)$ of $T_y(Y)$. We can easily check that this function $\mbox{sign}_{X\times Y}$ is an orientation of $X\times Y$, and thus, $X\times Y$ is oriented by this orientation.

Suppose that $X\subset \mathbb{R}^N$ is an $n$-dimensional $C^k$ oriented manifold, and $Y\subset X$ is an $(n-1)$-dimensional submanifold. Choose any orientation $\mbox{sign}_X$ of $X$. Suppose also that there exists a $C^{k-1}$ function $n:Y\to \mathbb{R}^N$ such that $n(x)\in T_x(X)$, $\|n(x)\|=1$, and $n(x)\cdot v=0$ for all $v\in T_x(Y)$. Then, there exists an orientation $\mbox{sign}_Y$ of $Y$ such that
\[\mbox{sign}_Y(x,v_1,...,v_{n-1})=\mbox{sign}_X(x,n(x),v_1,...,v_{n-1}).\]
Hence, $Y$ is an oriented manifold.

Next, suppose that $X\subset \mathbb{R}^N$ is a $C^k$ manifold with boundary. Choose any $x\in \partial X$. Then, there exists $p\in T_x(X)$ such that $\|p\|=1$ and $p\cdot v=0$ for all $v\in T_x(\partial X)$. Choose any local parametrization $\phi$ around $x$, and suppose that $p=D\phi_x(q)$. Then, $q_n\neq 0$. Define $n(x)=p$ if $q_n<0$ and $n(x)=-p$ if $q_n>0$. This function $n:\partial X\to \mathbb{R}^N$ is called the {\bf outward normal vector field} of $\partial X$. We can show that $n(x)$ is $C^{k-1}$, and thus, we can conclude that $\partial X$ is oriented by the following sign-preserving function:
\[\mbox{sign}_{\partial X}(x,v_1,...,v_{n-1})=\mbox{sign}_X(x,n(x),v_1,...,v_{n-1}).\]
Suppose that $X$ is a $C^k$ oriented manifold with boundary, $Y$ is a $C^k$ oriented manifold without boundary, $Z$ is a $C^k$ oriented submanifold of $Y$ without boundary, $f:X\to Y$ is a $C^k$ function such that both $f$ and $\partial f$ are transversal to $Z$, and $S=f^{-1}(Z)$ is nonempty. Then, $S$ is a manifold with boundary. Let $\ell=\mbox{codim}~Z$. Choose any $x\in f^{-1}(Z)$ and suppose that $H$ is an $\ell$-dimensional linear subspace of $T_x(X)$ such that
\[df_x(H)\oplus T_{f(x)}(Z)=T_{f(x)}(Y).\]
Then, we can show that
\[H\oplus T_x(S)=T_x(X).\]
Choose any basis $(v_1,...,v_{\ell})$ of $H$, and define $w_i=df_x(v_i)$. Choose any basis $(w_{\ell+1},...,w_m)$ of $T_{f(x)}(Z)$ such that $\mbox{sign}_Z(f(x),w_{\ell+1},...,w_m)=+1$. Define
\[h(v_1,...,v_{\ell})=\mbox{sign}_Y(f(x),w_1,...,w_m),\]
and for a basis $v_{\ell+1},...,v_n$ of $T_x(S)$,
\[\mbox{sign}_S(x,v_{\ell+1},...,v_n)=\mbox{sign}_X(x,v_1,...,v_n)h(v_1,...,v_{\ell}).\]
Then, we can show that $\mbox{sign}_S$ is an orientation. We call this orientation the {\bf preimage orientation}.

Finally, suppose that $X\subset \mathbb{R}^N$ is a $C^{\infty}$ oriented manifold. Consider the tangent bundle $T(X)$. Then, $T_{(x,v)}(T(X))=(T_x(X))^2$, and it is easy to show that there exists an orientation $\mbox{sign}_{T(X)}$ such that
\[\mbox{sign}_{T(X)}((x,v),(v_1,0),...,(v_n,0),(0,w_1),...,(0,w_n))=\mbox{sign}_X(x,v_1,...,v_n)\mbox{sign}_X(x,w_1,...,w_n).\]
Using this orientation, we can treat $T(X)$ as an oriented manifold.

\subsection{Classification of One-dimensional Manifolds}
Let
\[S^1=\{x\in \mathbb{R}^2|\|x\|=1\}.\]
Suppose that $X$ is a compact and connected $1$-dimensional $C^k$ manifold with/without boundary. Then, it is known that $X$ is $C^k$ diffeomorphic to either $S^1$ or $[0,1]$. In Guillemin and Pollack (1974), this classification theorem was proved using a Morse function. However, Morse function can be defined only if $k\ge 2$, and thus their proof cannot be applied if $k=1$. And actually, we need this result for $k=1$. Hence, we prove this theorem in the appendix.

It is known that every connected component of a manifold is open. Suppose that $X$ is a compact $1$-dimensional $C^k$ manifold with/without boundary. Because $X$ is compact, the number of connected components of $X$ is finite. Because each connected component $X_i$ is diffeomorphic to either $S^1$ or $[0,1]$, either $\partial X_i$ is empty or it contains $2$ points. In particular, $\partial X$ is a finite set that includes even number of points.

Moreover, suppose that $X$ is oriented. Then, for each connected component $X_i$ that is diffeomorphic to $[0,1]$, if $\partial X_i=\{a,b\}$, then $\mbox{sign}_{\partial X}(a)+\mbox{sign}_{\partial X}(b)=0$. This implies that
\[\sum_{x\in \partial X}\mbox{sign}_{\partial X}(x)=0.\]
This result is crucial for our main result.

\section{Result}
\subsection{Statement of the Result}
In this section, we fix $n\ge 2$, and let $A^T$ denote the transpose of the matrix $A$. We define $\mathbb{R}^n_+=\{x\in \mathbb{R}^n|x_i\ge 0\mbox{ for all }i\}$ and $\mathbb{R}^n_{++}=\{x\in \mathbb{R}^n|x_i>0\mbox{ for all }i\}$. Let $S=\{p\in \mathbb{R}^n_{++}|\|p\|=1\}$. It is known that $S$ is an $(n-1)$-dimensional $C^{\infty}$ manifold. Moreover, because $h(p)=p$ is a natural normal vector field on $S$, we can assume that $S$ is oriented.

Our main result is as follows.

\begin{theorem}
Let $f:\mathbb{R}^n_{++}\to \mathbb{R}^n$ be a continuous mapping such that the following five conditions hold.
\begin{enumerate}[{\rm (i)}]
\item $p\cdot f(p)=0$ for all $p\in S$, and $f(ap)=f(p)$ for all $a>0$ and $p\in S$.

\item There exists $s>0$ such that $f_i(p)>-s$ for all $p\in S$ and $i\in \{1,...,n\}$.

\item If $(p^k)$ is a sequence on $S$ such that $p^k\to p\notin S$ as $k\to \infty$, then $\|f(p^k)\|\to \infty$ as $k\to \infty$.

\item If $f(p^*)=0$, then $f$ is continuously differentiable around $p^*$.

\item If $f(p^*)=0$, then 
\[g(p^*)=\begin{vmatrix}
Df(p^*) & p^*\\
(p^*)^T & 0
\end{vmatrix}\neq 0.
\]
\end{enumerate}
Let $E^*=\{p\in S|f(p)=0\}$, and for each $p^*\in E^*$, define
\[\mathrm{index}(p^*)=\begin{cases}
+1 & \mbox{if }(-1)^ng(p^*)>0,\\
-1 & \mbox{if }(-1)^ng(p^*)<0.
\end{cases}\]
Then, $E^*$ is a finite set and
\[\sum_{p^*\in E^*}\mathrm{index}(p^*)=+1.\]
\end{theorem}

\vspace{12pt}
This result is actually the same as Proposition 5.6.1 of Mas-Colell (1985). However, the proof of this proposition includes significant gaps, and thus many readers cannot judge whether this proposition is correct. Hence, we provide a rigorous and gap-free proof for this theorem. This is the purpose of this paper.

\subsection{The Sketch of the Proof}
Because the proof of Theorem 1 is lengthy, providing a sketch of the proof may help the readers understand it. The most famous proof of such an index theorem is to apply the Poincar\'e-Hopf theorem. However, this theorem cannot be used for the present proof because the set $S$ is not compact. Fortunately, using assumption (iii), we can easily prove that zeroes of $f$ vanish around the corner of $S$. Thus, if we define a slightly smaller set $S_{\varepsilon}=\{p\in S|p_i\ge \varepsilon\}$, then this set will be compact, and all zeroes of $f$ are included in this set. However, this set is not a manifold. On the other hand, the interior $\tilde{S}_{\varepsilon}=\{p\in S|p_i>\varepsilon\}$ is a manifold. However, this set is not compact.

Therefore, it is difficult to use the Poincar\'e-Hopf theorem to prove Theorem 1. Hence, we consider using the homotopy invariance of the sum of indices. Specifically, the following function
\[f^q(p)=(1/q\cdot p)q-p\]
is a vector field on $S$, and there is a unique zero $q$ with the index $+1$. Define
\[F(t,p)=(1-t)f(p)+tf^q(p).\]
We want to prove the invariance of the sum of indices with respect to the homotopy defined above. The most important step in this proof is to verify that $F^{-1}(0)$ is a compact one-dimensional manifold.

Readers familiar with the discussion of transversality will notice that if the domain of $F$ is a manifold and $f$ is smooth, then by choosing $q$ appropriately, $F^{-1}(0)$ is a one-dimensional manifold. Therefore, if the domain of $F$ is $[0,1]\times \tilde{S}_{\varepsilon}$, then $X\equiv F^{-1}(0)$ is a one-dimensional manifold. However, because $\tilde{S}_{\varepsilon}$ is not compact, $X$ is not necessarily compact. On the other hand, if the domain of $F$ is $[0,1]\times S_{\varepsilon}$, then the domain is compact, and thus $Y\equiv F^{-1}(0)$ is compact. However, in this case, it is doubtful whether $Y$ is a manifold, because the domain of $F$ is not a manifold.

Hence, we need to prove that $F(t,p)\neq 0$ for any $t\in [0,1]$ and $p\in S_{\varepsilon}\setminus \tilde{S}_{\varepsilon}$. If we show this, then we have that $X=Y$. Because $Y$ is compact and $X$ is a one-dimensional manifold, it is a one-dimensional compact manifold. Using this fact and well-known results for computing indices, we can obtain the desired result. This is a sketch of our proof.

\subsection{Proof of Theorem 1}
We separate the proof into nine steps.

\begin{step}
The set $E^*$ is finite.
\end{step}

\begin{proof}[{\bf Proof of Step 1}]
Choose any $p^*\in E^*$, and let $A=Df(p^*)$ and $V=\{v\in \mathbb{R}^n|v\cdot p^*=0\}$. If $v\in V,\ v\neq 0$ and $Av=0$, then
\[\begin{pmatrix}
Df(p^*) & p^*\\
(p^*)^T & 0
\end{pmatrix}\begin{pmatrix}
v\\
0
\end{pmatrix}=0,
\]
which contradicts (v). Therefore, we have that $v\in V$ and $v\neq 0$ imply $Av\neq 0$. Hence,
\[m=\min\{\|Av\|/\|v\||v\in V,\ v\neq 0\}>0.\]
Because of the definition of the Fr\'echet derivative, there exists $\varepsilon>0$ such that if $v\in V$ and $0<\|v\|<\varepsilon$, then
\[\|f(p^*+v)-f(p^*)-Av\|<\frac{\|v\|m}{2}.\]
Because $f(p^*)=0$,
\begin{align*}
\|v\|m\le \|Av\|\le&~\|f(p^*+v)-f(p^*)-Av\|+\|f(p^*+v)\|\\
<&~\frac{\|v\|m}{2}+\|f(p^*+v)\|,
\end{align*}
which implies that $f(p^*+v)\neq 0$. Now, for any $q\in (p^*+V)\cap \mathbb{R}^n_{++}$, define $P(q)=q/\|q\|$. Then, $P$ is a $C^{\infty}$ mapping from $(p^*+V)\cap \mathbb{R}^n_{++}$ into $S$, and it is easy to show that $dP_{p^*}$ is the identity mapping. By the inverse function theorem, there exists $\delta>0$ such that $P^{-1}$ is defined, injective on $\{p\in S|\|p-p^*\|<\delta\}$, and if $p\in S$ and $\|p-p^*\|<\delta$, then $\|P^{-1}(p)-p^*\|<\varepsilon$. By (i), $p\in S, 0<\|p-p^*\|<\delta$ implies that $f(p)\neq 0$, and thus $E^*$ is a discrete set.

By (iii), there exists $\varepsilon>0$ such that if $p\in S$ and $p_i\ge \varepsilon$ for all $i\in \{1,...,n\}$, then $f(p)\neq 0$. This implies that $E^*$ is compact, and thus finite. This completes the proof of Step 1.
\end{proof}

\begin{step}
Let $(p^k)$ be a sequence on $S$ such that $p^k$ converges to $p\notin S$ and $\frac{1}{\|f(p^k)\|}f(p^k)$ converges to $z\in \mathbb{R}^n$ as $k\to \infty$. Then, $z\ge 0$, $z\neq 0$ and $p\cdot z=0$.
\end{step}

\begin{proof}[{\bf Proof of Step 2}]
By (ii), $f_i(p^k)\to -\infty$ is impossible, and thus by (iii), we have that $z_i\ge 0$. Because $\|z\|=1$, $z\neq 0$. Finally, suppose that $p_i>0$. By (i),
\[f_i(p^k)=-\frac{1}{p^k_i}\left(\sum_{j\neq i}p^k_jf_j(p^k)\right).\]
Because the right-hand side is bounded from above, we have that the sequence $(f_i(p^k))$ is bounded, and thus $z_i=0$, which implies that $p\cdot z=0$. This completes the proof of Step 2.
\end{proof}

\begin{step}
Choose any $q\neq 0$. Suppose that $A$ is an $n\times n$ matrix such that $v^TAv<0$ for all $v\in \mathbb{R}^n$ such that $v\cdot q=0$ and $v\neq 0$. Then,
\[(-1)^n\begin{vmatrix}
A & q\\
q^T & 0
\end{vmatrix}>0.\]
\end{step}

\begin{proof}[{\bf Proof of Step 3}]
Let $A_t=(1-t)A+tA^T$. By assumption, we have that $v^TA_tv<0$ for all $t\in [0,1]$ and $v\in\mathbb{R}^n$ such that $v\cdot q=0$ and $v\neq 0$. Suppose that there exists $t\in [0,1]$ such that
\[\begin{vmatrix}
A_t & q\\
q^T & 0
\end{vmatrix}=0,\]
then, there exist $v\in \mathbb{R}^n$ and $w\in \mathbb{R}$ such that $(v,w)\neq 0$ and
\[\begin{pmatrix}
A_t & q\\
q^T & 0
\end{pmatrix}
\begin{pmatrix}
v\\
w
\end{pmatrix}=0.\]
This implies that $q^Tv=0$, and thus $v\cdot q=0$. Therefore,
\[0=v^T0=v^TA_tv+v^Tqw=v^TA_tv,\]
which implies that $v=0$. However, this implies that $wq=0$, and thus $w=0$, which is a contradiction. By the intermediate value theorem, we have that
\[c(t)\equiv \begin{vmatrix}
A_t & q\\
q^T & 0
\end{vmatrix}\]
has a constant sign on $[0,1]$. If we set $t=\frac{1}{2}$, then $A_t$ is symmetric, and thus by Theorem 5 of Debreu (1952), $(-1)^nc(1/2)>0$. Therefore, $(-1)^nc(0)>0$, as desired. This completes the proof of Step 3.
\end{proof}

\begin{step}
Suppose that $p\in \mathbb{R}^n$ and $\|p\|=1$, and $A$ is an $n\times n$ matrix such that $Ap=0$ and
\[\begin{vmatrix}
A & p\\
p^T & 0
\end{vmatrix}\neq 0.\]
Let $x_1,...,x_{n-1}\in \mathbb{R}^n$ satisfy $p\cdot x_i=0$ and $\det(p,x_1,...,x_{n-1})>0$, and define $y_i=Ax_i$. Then,
\[\mathrm{sign}\begin{vmatrix}
A & p\\
p^T & 0
\end{vmatrix}=-\mathrm{sign}(\det(p,y_1,...,y_{n-1})).\]
\end{step}

\begin{proof}[{\bf Proof of Step 4}]
Let
\[B=\begin{pmatrix}
A & p \\
p^T & 0
\end{pmatrix}.\]
Then,
\[B\begin{pmatrix}
x_i\\
0
\end{pmatrix}=\begin{pmatrix}y_i\\
0
\end{pmatrix},\ B\begin{pmatrix}
p\\
0
\end{pmatrix}=\begin{pmatrix}
0\\
1
\end{pmatrix},\ B\begin{pmatrix}
0\\
1
\end{pmatrix}=\begin{pmatrix}
p\\
0
\end{pmatrix}.\]
Define
\[C=\begin{pmatrix}
0 & x_{11} & ... & x_{n-1,1} & p_1 \\
\vdots & \vdots & \ddots & \vdots & \vdots \\
0 & x_{1n} & ... & x_{n-1,n} & p_n \\
1 & 0 & ... & 0 & 0
\end{pmatrix}.\]
Then, $\det(C)<0$ and
\[BC=\begin{pmatrix}
p_1 & y_{11} & ... & y_{n-1,1} & 0 \\
\vdots & \vdots & \ddots & \vdots & \vdots \\
p_n & y_{1n} & ... & y_{n-1,n} & 0 \\
0 & 0 &... & 0 & 1
\end{pmatrix},\]
and thus
\[\det(B)=\det(p,y_1,...,y_{n-1})(\det(C))^{-1},\]
which implies that our claim is correct. This completes the proof of Step 4.
\end{proof}

Choose any $q\in S$, and define $f^q(p)=\frac{\|p\|}{p\cdot q}q-\frac{1}{\|p\|}p$. Then, $f^q$ is defined on $\mathbb{R}^n_+\setminus\{0\}$.

\begin{step}
$f^q$ satisfies $(\mathrm{i})$, $(\mathrm{ii})$, $(\mathrm{iv})$, $(\mathrm{v})$ of this theorem.
\end{step}

\begin{proof}[{\bf Proof of Step 5}]
Clearly, $p\cdot f^q(p)=0$ and $f^q(p)=f^q(ap)$ for all $a>0$, and thus (i) holds. Let $\bar{S}$ be the closure of $S$ in $\mathbb{R}^n$. Then, $f^q$ is defined on $\bar{S}$ and continuous, which implies that (ii) holds. Clearly, $f^q$ is $C^{\infty}$ and thus (iv) holds. It is obvious that $f^q(p)=0$ if and only if $p=aq$ for some $a>0$. Let $v\cdot q=0$ and $v\neq 0$, and define $p(t)=q+tv$. Then, $p(t)\cdot q=\|q\|^2=1$, and thus
\[\frac{d}{dt}f^q(p(t))=\frac{t\|v\|^2}{\|p(t)\|}q+\frac{t\|v\|^2}{2\|p(t)\|^3}p(t)-\frac{1}{\|p(t)\|}v,\]
which implies that
\[v^TDf^q(q)v=v\cdot \left(\left.\frac{d}{dt}f^q(p(t))\right|_{t=0}\right)=-\|v\|^2<0.\]
Therefore, $v\cdot q=0$ and $v\neq 0$ imply that $v^TDf^q(q)v<0$, and by Step 3,
\begin{equation}\label{EVAL}
(-1)^n\begin{vmatrix}
Df^q(q) & q\\
q^T & 0
\end{vmatrix}>0.
\end{equation}
Because $f^q(ap)=f^q(p)$ for all $a>0$, we have that $Df^q(aq)=a^{-1}Df^q(q)$, and thus (v) holds. This completes the proof of Step 5.
\end{proof}

Define
\[S_{\varepsilon}=\{p\in S|p_i\ge \varepsilon \mbox{ for all }i\in \{1,...,n\}\},\]
\[\tilde{S}_{\varepsilon}=\{p\in S|p_i>\varepsilon \mbox{ for all }i\in \{1,...,n\}\},\]
\[\partial S_{\varepsilon}=S_{\varepsilon}\setminus \tilde{S}_{\varepsilon}.\]
Note that, $S_{\varepsilon}$ and $\partial S_{\varepsilon}$ are not manifolds. Choose any $\varepsilon'>0$ such that $\tilde{S}_{\varepsilon'}\neq \emptyset$.

\begin{step}
If $\varepsilon>0$ is sufficiently small, then $\tilde{S}_{\varepsilon}\neq \emptyset$ and
\[(1-t)f(p)+tf^q(p)\neq 0\]
for all $t\in [0,1]$, $q\in S_{\varepsilon'}$, and $p\in \partial S_{\varepsilon}$.
\end{step}

\begin{proof}[{\bf Proof of Step 6}]
Suppose not. Then, there exist sequences $(t^k)$ on $[0,1]$, $(p^k)$ on $S$, and $(q^k)$ on $S_{\varepsilon'}$ such that $p^k\to p\notin S$ as $k\to \infty$ and
\[(1-t^k)f(p^k)+t^kf^{q^k}(p^k)=0\]
for all $k$. Let $z^k=\frac{1}{\|f(p^k)\|}f(p^k)$. By taking subsequences, we can assume that $z^k\to z$ as $k\to \infty$. By Step 2, $z\ge 0$, $z\neq 0$ and $p\cdot z=0$. Choose any $i$ such that $z_i>0$. Then, $p_i=0$, and $\min_{q\in S_{\varepsilon'}}f^q_i(p)>0$. This implies that $\min_{q\in S_{\varepsilon'}}f^q_i(p^k)>0$ for sufficiently large $k$. Moreover, because $z_i>0$, $f_i(p^k)>0$ for sufficiently large $k$. Thus, for such a $k$,
\[(1-t^k)f_i(p^k)+t^kf_i^{q^k}(p^k)>0,\]
which is a contradiction. This completes the proof of Step 6.
\end{proof}

Hereafter, we fix $\varepsilon>0$ that satisfies the requirements in Step 6, $E^*\subset S_{\varepsilon}$, and $\varepsilon<\varepsilon'$.

Choose any $\eta>0$ such that if $p^*\in E^*$, then $f$ is $C^1$ on $\{p\in \mathbb{R}^n_{++}|\|p-p^*\|<4\eta\}$, and $\{p\in \mathbb{R}^n_{++}|f(p)=0,\ \|p-p^*\|<4\eta\}\subset\{ap^*|a>0\}$. By Step 1, $E^*$ is finite, and thus such an $\eta$ exists. Define
\[\hat{S}_{\varepsilon,a}=\{p\in S_{\varepsilon}|\forall p^*\in E^*,\ \|p-p^*\|\ge a\}.\]
Then, $\hat{S}_{\varepsilon,a}$ is compact. We assume without loss of generality that $\hat{S}_{\varepsilon,2\eta}$ is nonempty. Define
\[\varphi(p)=\begin{cases}
Ce^{-\frac{1}{1-\|p\|^2}} & \mbox{if }\|p\|<1,\\
0 & \mbox{if }\|p\|\ge 1,
\end{cases}\]
where $C>0$ is chosen as
\[\int_{\mathbb{R}^n}\varphi(p)dp=1.\]
For $\delta>0$, define
\[\varphi_{\delta}(p)=\delta^n\varphi(p/\delta),\]
and
\[f_{\delta}(p)=\int_{\mathbb{R}^n} f(p-q)\varphi_{\delta}(q)dq.\]
It is known that $f_{\delta}$ uniformly converges to $f$ on the compact set $S_{\varepsilon}$ as $\delta\downarrow 0$. Moreover, $f_{\delta}$ is $C^{\infty}$.\footnote{See Appendix C.5 of Evans (2010).} Choose any $C^{\infty}$ function $\xi:\mathbb{R}\to [0,1]$ such that $\xi^{-1}(1)=]-\infty,1]$ and $\xi^{-1}(0)=[2,+\infty[$, and for any $p^*\in E^*$, let
\[t_{p^*}(p)=\xi(\eta^{-1}\|p-p^*\|).\]
For every $p\in S_{\varepsilon}$, define\footnote{If necessary, we can define $\tilde{f}_{\delta}(p)=\tilde{f}_{\delta}(p/\|p\|)$, and extend the domain of $\tilde{f}_{\delta}$ to a cone generated by $S_{\varepsilon}$, which includes an open neighborhood of $E^*$.}
\[\hat{f}_{\delta}(p)=f_{\delta}(p)-(f_{\delta}(p)\cdot p)p,\]
\[\tilde{f}_{\delta}(p)=\left(1-\sum_{p^*\in E^*}t_{p^*}(p)\right)\hat{f}_{\delta}(p)+\sum_{p^*\in E^*}t_{p^*}(p)f(p).\]
Because $f_{\delta}$ converges to $f$ uniformly on $S_{\varepsilon}$ as $\delta\downarrow 0$, there exists $\delta>0$ such that
\begin{enumerate}[1)]
\item For any $p\in S_{\varepsilon}$, $\tilde{f}_{\delta}(p)=0$ if and only if $p\in E^*$.

\item If $\tilde{f}_{\delta}(p)=0$, then $\tilde{f}_{\delta}=f$ on some neighborhood of $p$.

\item $p\cdot \tilde{f}_{\delta}(p)=0$ for all $p\in S_{\varepsilon}$.

\item For all $p\in \partial S_{\varepsilon}$, $q\in S_{\varepsilon'}$, and $t\in [0,1]$, $(1-t)\tilde{f}_{\delta}(p)+tf^q(p)\neq 0$.

\item $\tilde{f}_{\delta}$ is $C^1$ on $S_{\varepsilon}$.
\end{enumerate}
Clearly, if $\tilde{f}_{\delta}$ satisfies the claim of this theorem, then $f$ also satisfies the same claim. Hence, hereafter we implicitly use $\tilde{f}_{\delta}$ instead of $f$, and thus assume that $f$ is $C^1$ on $S_{\varepsilon}$.

Recall the definition of the tangent bundle:
\[T(\tilde{S}_{\varepsilon})=\{(p,v)|p\in \tilde{S}_{\varepsilon},\ v\in T_p(\tilde{S}_{\varepsilon})\}.\]
Define a function $F:[0,1]\times S_{\varepsilon}\times \tilde{S}_{\varepsilon'}\to T(\tilde{S}_{\varepsilon})$ as
\[F(t,p,q)=(p,(1-t)f(p)+tf^q(p)),\]
and $\tilde{F}$ as the restriction of $F$ into $[0,1]\times \tilde{S}_{\varepsilon}\times \tilde{S}_{\varepsilon'}$. Moreover, we define
\[G_q(t,p)=F(t,p,q),\ \tilde{G}_q(t,p)=\tilde{F}(t,p,q).\]
Finally, for any $(t,p,q)\in [0,1]\times \tilde{S}_{\varepsilon}\times \tilde{S}_{\varepsilon'}$, define $\pi(t,p,q)=q$. Because $[0,1]\times \tilde{S}_{\varepsilon}\times \tilde{S}_{\varepsilon'}$ is a $C^{\infty}$ manifold with boundary and $\pi$ is $C^{\infty}$, we can apply Sard's theorem to $\pi$ and $\partial \pi$, and thus there exists $q\in \tilde{S}_{\varepsilon'}$ such that it is a regular value of $\pi$ and $\partial \pi$.

\begin{step}
If $q$ is a regular value of both $\pi$ and $\partial \pi$, then both $\tilde{G}_q$ and $\partial \tilde{G}_q$ are transversal to $\tilde{S}_{\varepsilon}\times \{0\}$.
\end{step}

\begin{proof}[{\bf Proof of Step 7}]
First, we verify that $\tilde{F}$ and $\partial \tilde{F}$ are transversal to $\tilde{S}_{\varepsilon}\times \{0\}$. Note that
\[T_p(\tilde{S}_{\varepsilon})=\{v\in \mathbb{R}^n|v\cdot p=0\},\]
\[T_{(p,0)}(\tilde{S}_{\varepsilon}\times \{0\})=T_p(\tilde{S}_{\varepsilon})\times \{0\},\]
\[T_{(p,0)}(T(\tilde{S}_{\varepsilon}))=T_p(\tilde{S}_{\varepsilon})\times T_p(\tilde{S}_{\varepsilon}).\]
Suppose that $\tilde{F}(t,p,q)=(p,0)$. Let $c:]-\varepsilon'',\varepsilon''[\to \tilde{S}_{\varepsilon}\times \tilde{S}_{\varepsilon'}$ be a $C^1$ function such that $c(0)=(p,q)$ and $c'(0)=(v,w)$, and define $d(s)=\tilde{F}(t,c(s))$. Then, $v\cdot p=w\cdot q=0$, and thus
\[d'(0)=(v,(1-t)df_p(v)+tdf_p^q(v)+(t/p\cdot q)[w-(w\cdot p/p\cdot q)q]).\]
Because the mapping $w\mapsto w-(w\cdot p/p\cdot q)q$ is a bijection from $T_q(\tilde{S}_{\varepsilon'})$ into $T_p(\tilde{S}_{\varepsilon})$, both $d\tilde{F}_{(t,p,q)}$ and $d(\partial \tilde{F})_{(t,p,q)}$ are surjective whenever $t\neq 0$. Suppose that $t=0$. By the same argument as in the proof of Step 1, we have that $df_p(v)=0$ if and only if $v=0$, and thus for every $w\in T_p(\tilde{S}_{\varepsilon})$, there exists $v\in T_p(\tilde{S}_{\varepsilon})$ such that $df_p(v)=w$. Therefore, in this case,
\[\mbox{Im}(d\tilde{F}_{(0,p,q)})+T_{(p,0)}(\tilde{S}_{\varepsilon}\times \{0\})=T_{(p,0)}(T(\tilde{S}_{\varepsilon})),\]
\[\mbox{Im}(d(\partial \tilde{F})_{(0,p,q)})+T_{(p,0)}(\tilde{S}_{\varepsilon}\times \{0\})=T_{(p,0)}(T(\tilde{S}_{\varepsilon})),\]
which implies that $\tilde{F}$ and $\partial \tilde{F}$ are transversal to $\tilde{S}_{\varepsilon}\times \{0\}$. As in the proof of the transversality theorem in Section 2.3 of Guillemin and Pollack (1974), if both $\pi$ and $\partial \pi$ have a common regular value $q$, then $\tilde{G}_q$ and $\partial \tilde{G}_q$ are transversal to $\tilde{S}_{\varepsilon}\times \{0\}$. This completes the proof of Step 7.
\end{proof}

\begin{step}
Define $X\equiv \tilde{G}_q^{-1}(\tilde{S}_{\varepsilon}\times \{0\})$. Then, $X$ is a compact $1$-dimensional $C^1$ manifold with boundary, where $\partial X=\{(0,p^*)|p^*\in E^*\}\cup \{(1,q)\}$.\footnote{Note that, because $\varepsilon<\varepsilon'$, we have that $q\in \tilde{S}_{\varepsilon}$.}
\end{step}

\begin{proof}[{\bf Proof of Step 8}]
Because of Step 6,\footnote{Actually, we use property 4) of $\tilde{f}_{\delta}$.} we have that $X=G_q^{-1}(S_{\varepsilon}\times \{0\})$. Because $G_q$ is continuous, $X$ is closed. Moreover, because $X\subset [0,1]\times S_{\varepsilon}$, $X$ is compact. On the other hand, because $\tilde{G}_q$ and $\partial \tilde{G}_q$ are transversal to $\tilde{S}_{\varepsilon}\times \{0\}$, we have that $X$ is $1$-dimensional $C^1$ manifold with boundary, where $\partial X=X\cap (\{0,1\}\times \tilde{S}_{\varepsilon})$. Then,
\[X\cap (\{0,1\}\times \tilde{S}_{\varepsilon})=\{(0,p^*)|p^*\in E^*\}\cup \{(1,q)\},\]
and our claim is correct. This completes the proof of Step 8.
\end{proof}

\begin{step}
Suppose that $\partial X$ is oriented by the orientation induced by the outward normal vector field. Then, $\mathrm{sign}_{\partial X}(1,q)=(-1)^{n-1}$, and for all $p^*\in E^*$,
\[\mathrm{sign}_{\partial X}(0,p^*)=(-1)^n\mathrm{index}(p^*).\]
\end{step}

\begin{proof}[{\bf Proof of Step 9}]
Let $n(p)=p$. Then, $n(p)$ is a normal vector field on $\tilde{S}_{\varepsilon}$. Therefore, using this vector field, we can obtain an orientation $\mbox{sign}_{\tilde{S}_{\varepsilon}}$: that is, for a linearly independent family $v_1,...,v_{n-1}\in T_p(\tilde{S}_{\varepsilon})$,
\[\mbox{sign}_{\tilde{S}_{\varepsilon}}(p,v_1,...,v_{n-1})=\mbox{sign}(\det(p,v_1,...,v_{n-1})).\]
Note that $[0,1]$ can be oriented by the usual function
\[\mbox{sign}_{[0,1]}(t,a)=\mbox{sign}(a).\]
Therefore, $[0,1]\times \tilde{S}_{\varepsilon}$ is oriented by the product orientation. As we argued in Subsection 2.5, the tangent bundle $T(\tilde{S}_{\varepsilon})$ is oriented by the following function $\mbox{sign}_{T(\tilde{S}_{\varepsilon})}$:\footnote{It is easy to verify this claim is true if $v_i=(v_i',0)$ and $w_i=(0,w_i')$ for all $i$, and we can easily extend this result to general cases.}
\begin{align*}
&~\mbox{sign}_{T(\tilde{S}_{\varepsilon})}((p,v),v_1,...,v_{n-1},w_1,...,w_{n-1})\\
=&~\mbox{sign}(\det((p,0),v_1,...,v_{n-1},(0,p),w_1,...,w_{n-1})).
\end{align*}
Therefore, $X$ is also oriented by the preimage orientation. Suppose that $t\in \{0,1\}$, $(t,p)\in X$, and $(1,v)\in T_{(t,p)}(X)$. Choose $v_1,...,v_{n-1}\in T_p(\tilde{S}_{\varepsilon})$ such that $\det(p,v_1,...,v_{n-1})>0$. Because $\partial \tilde{G}_q$ is transversal to $\tilde{S}_{\varepsilon}\times \{0\}$, there exists $v_1',...,v_{n-1}'\in T_p(\tilde{S}_{\varepsilon})$ such that $\det(p,v_1',...,v_{n-1}')>0$ and for $(v_i',w_i')=d(\tilde{G}_q)_{(t,p)}(0,v_i')$, $x=((v_1',w_1'),...,(v_{n-1}',w_{n-1}'),(v_1,0),...,(v_{n-1},0))$ consists of a basis of $T_{\tilde{G}_q(t,p)}(T(\tilde{S}_{\varepsilon}))$.\footnote{Note that $d(\tilde{G}_q)_{(t,p)}(0,v')=(v',w')$ by the definition of $\tilde{G}$.} By the definition of the preimage orientation,
\begin{align*}
&~\mbox{sign}_X((t,p),(1,v))\\
=&~\mbox{sign}_{T(\tilde{S}_{\varepsilon})}((p,0),x)\times \mbox{sign}_{[0,1]\times \tilde{S}_{\varepsilon}}((t,p),(0,v_1'),...,(0,v_{n-1}'),(1,v)).
\end{align*}
Note that, for $r_1,...,r_n\in T_{(t,p)}([0,1]\times \tilde{S}_{\varepsilon})$,
\[\mbox{sign}_{[0,1]\times \tilde{S}_{\varepsilon}}((t,p),r_1,...,r_n)=\mbox{sign}(\det(r_1,(0,p),r_2,...,r_n)),\]
and thus, we have that
\begin{align*}
&~\mbox{sign}_{[0,1]\times \tilde{S}_{\varepsilon}}((t,p),(0,v_1'),...,(0,v_{n-1}'),(1,v))\\
=&~\mbox{sign}((-1)^{n+1}\det(p,v_1',...,v_{n-1}'))=(-1)^{n+1}.
\end{align*}
On the other hand,
\begin{align*}
\mbox{sign}_{T(\tilde{S}_{\varepsilon})}((p,0),x)=&~\mbox{sign}((-1)^{n^2+1}\det(p,w_1',...,w_{n-1}')\det(p,v_1,...,v_{n-1}))\\
=&~(-1)^{n+1}\mbox{sign}(\det(p,w_1',...,w_{n-1}')).
\end{align*}
Hence, we obtain that
\[\mbox{sign}_X((t,p),(1,v))=\mbox{sign}(\det(p,w_1',...,w_{n-1}')).\]
Suppose that $t=1$. Then, $p=q$. Because $w_i'=Df^q(q)v_i'$, by Step 4 and (\ref{EVAL}), we have that $\mbox{sign}_X((1,q),(1,v))=(-1)^{n-1}$, and thus $\mbox{sign}_{\partial X}(1,q)=(-1)^{n-1}$.

Next, suppose that $t=0$. Then, $p=p^*$ for some $p^*\in E^*$. Because $w_i'=Df(p)v_i'$, by Step 4, we have that $\mbox{sign}_X((0,p^*),(1,v))=(-1)^{n-1}\mbox{index}(p^*)$. Because $(1,v)$ is an inward vector at $T_{(0,p^*)}(X)$, we have that
\[\mbox{sign}_{\partial X}(0,p^*)=\mbox{sign}_X((0,p^*),(-1,-v))=(-1)^n\mbox{index}(p^*).\]
This completes the proof of Step 9.
\end{proof}

Because any compact $1$-dimensional $C^1$ manifold is a finite union of connected components, and each of them is diffeomorphic to either $S^1$ or $[0,1]$, we have that
\[\sum_{(t,p)\in \partial X}\mbox{sign}_{\partial X}(t,p)=0.\]
By Step 9,
\[\sum_{p^*\in E^*}\mbox{index}(p^*)=1,\]
as desired. This completes the proof of Theorem 1.

\section{Remarks}
Our Theorem 1 is, at least formally, different from Proposition 5.6.1 of Mas-Colell (1985). That is, Mas-Colell treated a function $f$ whose domain is $S$, whereas, in our Theorem 1, we treat a function $f$ defined on $\mathbb{R}^n_{++}$ that is homogeneous of degree zero. Moreover, the definitions of $\mbox{index}(p^*)$ are slightly different. However, readers can check that these differences are superficial. First, the restriction of a function treated in Theorem 1 into $S$ satisfies all the assumptions of Mas-Colell's proposition. Conversely, a function treated in Mas-Colell's proposition can be extended to $\mathbb{R}^n_{++}$ by the following formula $f(p)=f(p/\|p\|)$, and this extended function satisfies all requirements of our Theorem 1. Furthermore, Mas-Colell mentioned in Section 5.3 of his book that his definition of $\mbox{index}(p^*)$ is equivalent to our definition. Therefore, actually, our Theorem 1 is equivalent to Mas-Colell's proposition.

This result is closely related to the Poincar\'e-Hopf index theorem. This theorem claims for every smooth vector field in a compact smooth manifold, the sum of the ``indices'' of zeroes are the same as the Euler characteristic of this manifold. This theorem is often explained as the existence theorem of ``hair whorl.'' Let us look at the human head as a manifold and consider the direction of hair growth as a vector field. The human head can be seen as the upper hemisphere in $\mathbb{R}^3$, and its Euler characteristic is $1$. Therefore, there must be a zero with a positive index. This is the hair whorl. From this explanation, this theorem is also called the ``hairy ball theorem.''

However, upon closer examination, we find something strange. What would happen to the head of an alopecia areata patient? In this case, if we remove the part of the head that is bald from the upper hemisphere, a new manifold is created. This manifold is diffeomorphic to the compact domain of a large disc from which a small disc is removed. It is easy to see that the Euler characteristic of this manifold is $0$, because the rotating vector field has no zero. Hence, there must either be no hair whorl or an even number of hair whorls, although we have never heard of a hair whorl disappearing caused by alopecia areata. In other words, if a region is removed by alopecia areata, the Poincar\'e-Hopf index theorem does not appear to hold.

The reason for this lies in a hidden assumption in the Poincar\'e-Hopf index theorem. If the manifold has a boundary, then this theorem requires that the vector field is outward on that boundary. Normal human head hair satisfies this condition, and thus this theorem is applicable. However, this condition could not be satisfied on the boundary of the hair loss area caused by alopecia areata, and thus this theorem is not applicable.

With this in mind, let us look at our Theorem 1. Unfortunately, $S$ is not compact, and thus the index theorem is not applicable. On the other hand, $S_{\varepsilon}$ is compact. Although there is a problem that $S_{\varepsilon}$ is not a manifold, we ignore this problem and check whether the index theorem is applicable. Then, we encounter the problem that we indicated above. That is, our vector field $f$ is not necessarily outward on $\partial S_{\varepsilon}$. Therefore, the index theorem is still inapplicable.

To solve this problem, Mas-Colell prepared a function $f^q$, to which the conclusion of the index theorem holds. It is known that the sum of indices is a homotopy invariant. Because $f^q$ is homotopic to $f$, the sum of indices of $f$ is equal to that of $f^q$. However, $f^q$ has a unique zero $q$, whose index is $+1$. Hence, the sum of indices of $f$ is also $+1$. This is the `first' proof of this theorem provided by Mas-Colell. Curiously, Mas-Colell made the `second' proof of this theorem, but this proof is profound and unfamiliar.

There are a number of immediately visible problems with this proof. First, $S_{\varepsilon}$ is not a manifold. Although $\tilde{S}_{\varepsilon}$ is a manifold, it is not compact. If we change the definition of $S_{\varepsilon}$, then it may possibly be a manifold with boundary. However, even in this case, $[0,1]\times S_{\varepsilon}$ is not a manifold with boundary. As we have already seen in the proof, to show that the sum of indices is homotopy invariant, we need to use the transversality of $G_q(t,p)$ to prove that $X$ is a one-dimensional manifold and that the sum of orientations of its boundary is $0$. However, for the reasons mentioned above, the domain of $G_q$ is not a manifold, and thus it is unclear whether the result of transversality is applicable. Step 6 is essential for solving this problem. In fact, Mas-Colell himself proved Step 6, but he did not explain why it is needed at all.

Second, to apply the transversality result to $G_q$, the range must be a manifold. However, because $f$ is a vector field, the range of this function is not a manifold. Therefore, it is necessary to make the range of $G_q$ a manifold artificially by introducing the notion of the tangent bundle, and then apply the transversality theorem to $F$ and show that there exists $q$ such that $G_q$ is transversal. However, Mas-Colell did not mention the tangent bundle in his proof.

Hence, Mas-Colell's proof has a number of significant gaps. We think that our proof fulfills these gaps completely, and thus it is worthwhile to many economists that want to use the index theorem.

\appendix
\section{Proof of the Classification Theorem}
The classification theorem for one-dimensional manifolds can be found in Guillemin and Pollack (1974). However, as already mentioned, their proof uses the Morse function. Since the definition of the Morse function involves second-order derivatives, this argument cannot be applied to $C^1$ manifolds.

Milnor (1965) provided another proof using arc length. However, this proof is too difficult to understand. On the other hand, Exercise 1.4.9 in Hirsch (1976) asks us to prove this theorem. However, no answer is provided there. Moreover, Hirsch stated in the preface to his book that the exercises contain open problems. Therefore, the existence of this exercise does not solve this problem in any way. Gale (1987) treated this classification theorem. However, he only proved that there is a function that is not a diffeomorphism but only a homeomorphism. Thus, this cannot be used for our purpose.

In fact, as of November 2022, using Google search, we found a document that proved this theorem rigorously on the server of the University of Lisbon. Upon investigation, we found that the document was written by Gustavo Granja, who provided it for his student. However, there is no signature of the author inside the document, and thus this monograph is not in a position to be added to the reference list. In addition, since it is a lecture document, it could disappear at any time because of Granja's transfer or other reasons. Therefore, in this section, we reconstruct Granja's proof only as we need it.

Recall the definition of $S^1$:
\[S^1=\{x\in \mathbb{R}^2|\|x\|=1\}.\]
Suppose that $k\in \mathbb{N}\cup\{0,\infty\}$ is given, and suppose that $X$ is a $1$-dimensional compact and connected $C^k$ manifold. Our goal is to show the following theorem.

\begin{theorem}
$X$ is $C^k$ diffeomorphic to either $S^1$ or $[0,1]$.
\end{theorem}

\begin{proof}
In this proof, we call a function $f:W\to X$ a {\bf local parametrization of $X$} if and only if it is a $C^k$ diffeomorphism from $W$ onto $f(W)$, where either $W=(a,b)$ or $W=[a,b)$ for some $a,b\in \mathbb{R}$ with $a<b$ and $f(W)$ is open in $X$. By definition, we can easily check that $W=[a,b)$ means that $f(a)\in \partial X$.

Let $\varphi:U\to X$ and $\psi:V\to X$ be given $C^k$ functions such that $\varphi^{-1}$ and $\psi^{-1}$ are local parametrizations whose domains are $I$ and $J$, respectively, and $U,V\subset X$ satisfies
\begin{equation}\label{REST}
U\cap V\neq \emptyset,\ U\setminus V\neq \emptyset, V\setminus U\neq \emptyset.
\end{equation}
If a pair $(\varphi,\psi)$ of functions satisfies the above requirements, then we say that $(\varphi,\psi)$ satisfies {\bf condition A}.

We separate the proof into nine steps. In Steps 1-6, we assume that $(\varphi,\psi)$ satisfies condition A. By definition, $I,J\subset \mathbb{R}$ are bounded intervals.\footnote{We call a subset of $\mathbb{R}$ an {\bf interval} if this set is a convex set that contains at least two points. If $I$ is a bounded interval whose closure is $[a,b]$, then $a,b$ are called {\bf endpoints} of $I$.} Define $\alpha:\varphi(U\cap V)\to \psi(U\cap V)$ as $\alpha=\psi\circ \varphi^{-1}$. It is easy to show that $\varphi(U\cap V)$ is a (possibly infinitely many) disjoint union of relatively open separated intervals of I.

\setcounter{step}{0}
\begin{step}
Suppose that $K$ is a connected component of $\varphi(U\cap V)$ and $a$ is an endpoint of $K$. Then, $a\notin K$ and if $a\in I$, then
\[c\equiv \lim_{t\to a,\ t\in K}\alpha(t)\notin J,\]
and $c$ is an endpoint of $J$.
\end{step}

\begin{proof}[{\bf Proof of Step 1}]
Because $J$ is bounded, $\alpha(t)$ is bounded on $K$. Moreover, because $\alpha$ is a homeomorphism defined on $K$, it is monotone, which implies that $c$ is defined and is a real number.

Because $K$ is relatively open in $I$, if $a$ is not an endpoint of $I$, then $a\notin K$. Suppose that $a$ is an endpoint of $I$ and $a\in K$. Then, $a\in I$ and thus $I=[a,b)$ for some $b>a$, which implies that $\varphi^{-1}(a)\in \partial X$ and $\alpha(a)$ is also an endpoint of $J$. Let $\sup K=d$. Suppose that $d=b$. Then, $U\subset V$, which contradicts (\ref{REST}). Hence, $d<b$ and $a<d<b$. Let $\delta=\lim_{t\to d,\ t\in K}\alpha(t)$. If $\delta\notin J$, then $J=[\alpha(a),\delta)$, and thus $V\subset U$, which contradicts (\ref{REST}). Therefore, $\delta\in J$. Because $\varphi^{-1}$ and $\psi^{-1}$ are continuous, $\varphi^{-1}(d)=\psi^{-1}(\delta)$. Since $\psi$ is a homeomorphism, $U\cap V$ must include a neighborhood of $\psi^{-1}(\delta)$, and thus $K$ includes $[a,d+\varepsilon)$ for some $\varepsilon>0$, which is a contradiction. Thus, we conclude that $a\notin K$.

Next, suppose that $c\in J$. If $\varphi^{-1}(a)=\psi^{-1}(c)$, then $a\in K$, which is a contradiction. Therefore, $\varphi^{-1}(a)\neq \psi^{-1}(c)$, which contradicts the continuity of $\varphi^{-1}$ and $\psi^{-1}$. Hence, we have that $c\notin J$. This completes the proof.
\end{proof}

\begin{step}
$\varphi(U\cap V)$ consists of either one or two open intervals, and if $K$ is a connected component of $\varphi(U\cap V)$, then $K$ has at least one endpoint in common with $I$.
\end{step}

\begin{proof}[{\bf Proof of Step 2}]
Suppose that $K$ is a connected component of $\varphi(U\cap V)$. By Step 1, we have that $K$ is an open interval, and thus $K=(a,b)$ for some $a,b$ such that $a<b$. Let $c,d$ be endpoints of $I$ such that $c<d$. Then, $c\le a<b\le d$. If $c<a$ and $b<d$, then $V=\varphi^{-1}(K)\subset U$, which contradicts (\ref{REST}). Therefore, we have that either $c=a$ or $b=d$, which implies that $K$ has at least one endpoint in common with $I$. Therefore, there are at most two such $K$, which completes the proof of Step 2.
\end{proof}

\begin{step}
If $\varphi(U\cap V)$ consists of two connected components, then $I$ and $J$ are open, and $U\cup V$ is homeomorphic to $S^1$.
\end{step}

\begin{proof}[{\bf Proof of Step 3}]
Suppose that $a,b$ are endpoints of $I$ such that $a<b$. Then, there are two connected components $K_1=(a,c)$ and $K_2=(d,b)$ of $\varphi(U\cap V)$ such that $c\le d$. Define
\[e=\lim_{t\uparrow c}\alpha(t),\ f=\lim_{t\downarrow d}\alpha(t).\]
By Step 1, we have that $e,f\notin J$, and $e,f$ are endpoints of $J$. If $e=f$, then $\alpha$ is monotone in $K_1\cup K_2$, and because $J$ is convex, there exists $r\in V\setminus U$ such that $\psi(r)=e=f$. However, this implies that $e,f\in J$, which is a contradiction. Therefore, $e\neq f$, which implies that $J$ is an open interval. By symmetrical arguments, we can show that $I$ is an open interval, and thus $I=(a,b)$.

Thus, either $e<f$ or $e>f$. Because arguments can be done symmetrically, we assume that $e>f$, and thus $J=(f,e)$. Then, $\alpha$ is increasing on both $K_i$, and $\alpha(K_1)=(g,e)$ and $\alpha(K_2)=(f,h)$ for
\[g=\lim_{t\downarrow a}\alpha(t),\ h=\lim_{t\uparrow b}\alpha(t).\]
Because $e>f$, we have that $g\ge h$. If $g=h$, then $(U\cup V)\setminus U$ is a singleton, and thus $U\cup V$ can be seen as a one-point compactification of $U$, which is homeomorphic to $S^1$. If $g>h$, then choose an increasing homeomorphism $\lambda:(0,\pi)\to (a,b)$ and $\mu:[\pi,2\pi]\to [h,g]$, and define
\[f(e^{i\theta})=\begin{cases}
\varphi^{-1}(\lambda(\theta)) & \mbox{if }0<\theta<\pi,\\
\psi^{-1}(\mu(\theta)) & \mbox{if }\pi\le \theta\le 2\pi.
\end{cases}\]
We can easily check that $f$ is a homeomorphism between $S^1$ and $U\cup V$. This completes the proof of Step 3.
\end{proof}

\begin{step}
Suppose that $\varphi(U\cap V)=K$ for an open interval $K$ whose endpoints are in common with $I$. Then, $I=[a,b)$ and $K=(a,b)$ for some $a,b$ with $a<b$, and $U\cap \partial X\neq \emptyset$. If $V\cap \partial X=\emptyset$, then $U\cup V$ is homeomorphic to $[0,1)$. If $V\cap \partial X\neq \emptyset$, then $U\cup V$ is homeomorphic to $[0,1]$.
\end{step}

\begin{proof}[{\bf Proof of Step 4}]
Suppose that $K=(a,b)$. Because of (\ref{REST}), we have that $I\neq K$, and thus $I=[a,b)$. Therefore, $\varphi^{-1}(a)\in \partial X$, and thus $U\cap \partial X=\{\varphi^{-1}(a)\}$. Let $c,d$ be endpoints of $\psi(U\cap V)$, and
\[c=\lim_{t\downarrow a}\alpha(t).\]
By Step 1, $c\notin J$ and $c$ is an endpoint of $J$. Because of (\ref{REST}), we have that $d\in J$, and thus $\psi^{-1}(d)\in V\setminus U$. Suppose that $e$ is another endpoint of $J$. Then, $J$ is $(e,c)$, $(c,e)$, or $[e,c)$. If $J=(e,c)$, then $J\cap \partial X=\emptyset$ and $e<d$. In this case, choose homeomorphisms $\lambda:[0,1/2)\to [a,b)$ and $\mu:[1/2,1)\to (e,d]$, and define
\[f(x)=\begin{cases}
\varphi^{-1}(\lambda(x)) & \mbox{if }0\le x<1/2,\\
\psi^{-1}(\mu(x)) & \mbox{if }1/2\le x<1.
\end{cases}\]
Then, we can easily check that $f$ is a homeomorphism between $[0,1)$ and $U\cup V$. The cases $J=(c,e)$ and $J=[e,c)$ can be similarly treated. This completes the proof.
\end{proof}

\begin{step}
Suppose that $\varphi(U\cap V)=K$ for an open interval $K$ and there exists an endpoint $b$ of $K$ such that it is not an endpoint of $I$. Then, $U\cup V$ is homeomorphic to $Y$, where $Y$ is $(0,1),[0,1)$ or $[0,1]$. Suppose that $f:Y\to (U\cup V)$ is the diffeomorphism. If $Y=(0,1)$, then $(U\cup V)\cap \partial X=\emptyset$. If $Y=[0,1)$, then $(U\cup V)\cap \partial X=\{f(0)\}$. If $Y=[0,1]$, then $(U\cup V)\cap \partial X=\{f(0),f(1)\}$.
\end{step}

\begin{proof}[{\bf Proof of Step 5}]
This step can be proved in almost the same manner as the proof of Step 4. Hence, we omit the proof.
\end{proof}

\begin{step}
If $U\cup V$ is homeomorphic to $S^1$ or $[0,1]$, then $U\cup V=X$.
\end{step}

\begin{proof}[{\bf Proof of Step 6}]
In this case, $U\cup V$ becomes a clopen set. Therefore, by the connectedness of $X$, we have that $U\cup V=X$. This completes the proof of Step 6.
\end{proof}

Choose any pair $(\varphi,\psi)$ that satisfies condition A. Suppose that there exist an open set $U'$ of $X$ and an extension $\varphi':U'\to \mathbb{R}$ of $\varphi$ such that the closure $\bar{U}$ of $U$ is included in $U'$ and $(\varphi',\psi)$ satisfies condition A. Then, we say that $(\varphi,\psi)$ satisfies {\bf condition B}.

In Steps 7-8, we suppose that $(\varphi,\psi)$ satisfies condition B and $\varphi':U'\to I'$ is fixed.

\begin{step}
Suppose that $k\ge 1$, and $\varphi(U\cap V)$ has two connected components. Then, $U\cup V$ is $C^k$ diffeomorphic to $S^1$.
\end{step}

\begin{proof}[{\bf Proof of Step 7}]
As in the proof of Step 3, we define $a,b,c,d,e,f,g,h$. Then, $I=(a,b)$, $\varphi(U\cap V)=(a,c)\cup (d,b)$,
\[\lim_{t\uparrow c}\alpha(t)=e,\ \lim_{t\downarrow d}\alpha(t)=f,\]
\[\lim_{t\downarrow a}\alpha(t)=g,\ \lim_{t\uparrow b}\alpha(t)=h,\]
and $e\neq f$. Because arguments can be done symmetrically, we assume that $e>f$, and thus $J=(f,e)$. Then, $\alpha$ is increasing on both $K_i$, and $\alpha(K_1)=(g,e)$ and $\alpha(K_2)=(f,h)$. If $g=h$, then $U\cup V=U\cup \{\psi(g)\}$. By Step 3, in this case $U\cup V$ is homeomorphic to $S^1$, and thus by Step 6, $U\cup V=X$, which implies that $V\subset X=U'$. This contradicts that $U',V$ satisfies (\ref{REST}). Therefore, we must have that $h<g$. Choose an increasing linear function $\lambda:[-\delta,\pi+\delta]\to [a-\varepsilon,b+\varepsilon]$ such that $\delta<\pi/2$, $\lambda(0)=a,\ \lambda(\pi)=b$, and $[a-\varepsilon,c)\cup (d,b+\varepsilon]\subset \varphi'(U'\cap V)$, and thus the domain of $\alpha$ can be extended to $[a-\varepsilon,c)\cup (d,b+\varepsilon]$. Define $\beta=\alpha\circ \lambda$. Choose $p,q\in (h,g)$ such that
\[\beta(\pi+\delta/2)<p<q<\beta(-\delta/2).\]
Because $\beta(\pi+\delta)\le \beta(-\delta)$, we have that such $p,q$ exist. Let $\mu(\theta)=((2\pi-\theta)p+(\theta-\pi)q)/\pi$ on $[\pi,2\pi]$, and choose any $C^{\infty}$ function $\nu:[\pi,2\pi]\to [0,1]$ such that $\nu^{-1}(0)=[\pi,\pi+\delta/2]\cup [2\pi-\delta/2,2\pi]$, $\nu^{-1}(1)=[\pi+\delta,2\pi-\delta]$, and $\nu'(\theta)\neq 0$ if $\theta\in (\pi+\delta/2,\pi+\delta)\cup (2\pi-\delta,2\pi-\delta/2)$, and define
\[g(\theta)=(1-\nu(\theta))(\beta(\theta)+\beta(\theta-2\pi))+\nu(\theta)\mu(\theta),\]
where we define $\beta(\theta)=0$ if $\theta\notin [-\delta,\pi+\delta]$. Then, we can check that $g$ is $C^k$ on $[\pi,2\pi]$, $g'(\theta)>0$ for all $\theta\in [\pi,2\pi]$, and $g\equiv \beta$ on some neighborhood of $\{\pi,2\pi\}$. Define 
\[f(e^{i\theta})=\begin{cases}
\varphi^{-1}(\lambda(\theta)) & \mbox{if } 0<\theta<\pi,\\
\psi^{-1}(g(\theta)) & \mbox{if } \pi\le \theta\le 2\pi.
\end{cases}\]
Then, $f$ is a $C^k$ diffeomorphism between $S^1$ and $U\cup V$. This completes the proof of Step 7.
\end{proof}

\begin{step}
Suppose that $k\ge 1$, and $\varphi(U\cap V)$ is an open interval. Then, ``homeomorphic'' that appears in Steps 4-5 can be replaced with ``$C^k$ diffeomorphic.''
\end{step}

\begin{proof}[{\bf Proof of Step 8}]
We treat only the case in which $I,J$ are open intervals and $\varphi(U\cap V)=K$ for an open interval $K$, because the remaining cases can be treated in almost the same manner. It is easy to show that there exists a $C^{\infty}$ diffeomorphism $g_1:I\to (0,1)$ and $g_2:J\to (1,2)$ such that $g_1(\varphi(U\cap V))=(1/2,1)$ and $g_2(\psi(U\cap V))=(1,3/2)$. Therefore, without loss of generality, we can assume that $I=(0,1), J=(1,2)$, and $\varphi(U\cap V)=(1/2,1),\psi(U\cap V)=(1,3/2)$. This implies that $\alpha$ is increasing. Because $(\varphi,\psi)$ satisfies condition B, $\alpha$ can be extended to $(1/2,1+\varepsilon]$ for some $\varepsilon\in (0,1/2)$ such that $\alpha(1+\varepsilon)<2$. Choose any $C^{\infty}$ function $\lambda:(0,3/2)\to [0,1]$ such that $\lambda(x)=1$ if $x\le 1+\varepsilon/2$, $\lambda(x)=0$ if $x\ge 1+\varepsilon$, and $\lambda'(x)<0$ for all $x\in (1+\varepsilon/2,1+\varepsilon)$. Define
\[g(x)=\alpha(1+\varepsilon)+2(x-1)(2-\alpha(1+\varepsilon)),\]
and
\[f(x)=\begin{cases}
\varphi^{-1}(x) & \mbox{if }x<1,\\
\psi^{-1}(\lambda(x)\alpha(x)+(1-\lambda(x))g(x)) & \mbox{if }x\ge 1.
\end{cases}\]
Then, $f$ is a $C^k$ diffeomorphism between $(0,3/2)$ and $U\cup V$. This completes the proof of Step 8.
\end{proof}

\begin{step}
$X$ is $C^k$ diffeomorphic to $S^1$ or $[0,1]$.
\end{step}

\begin{proof}[{\bf Proof of Step 9}]
In this proof, let $\bar{W}$ denote the closure of $W$. Choose any open covering $(U_i)$ of $X$ such that there exists a $C^k$ diffeomorphism $\varphi_i:U_i\to A_i$ such that $A_i$ is either $(a,b)$ or $[a,b)$. Shrinking $U_i$ if needed, we can assume that $\varphi_i$ can be extended to some open set $U_i'$ of $X$ that includes $\bar{U}_i$, and $\varphi_i$ is a $C^k$ diffeomorphism defined on $U_i'$. Because $X$ is compact, there exists a finite subcovering $(U_{i_j})_{j=1}^n$. Define
\[V_j=U_{i_1}\cup ... \cup U_{i_j}.\]
Renumbering and removing some $U_{i_j}$ if needed, we can assume that
\[U_{i_j}\cap V_{j-1}\neq \emptyset,\ U_{i_j}\setminus V_{j-1}\neq \emptyset,\ V_{j-1}\setminus U_{i_j}\neq \emptyset.\]
Shrinking $U_{i_1}',...,U_{i_n}'$ if needed, we can assume without loss of generality that
\[V_{j-1}\setminus U_{i_j}'\neq \emptyset.\]
If $n=1$, then $X=U_{i_1}$ is $C^k$ diffeomorphic to either $[0,1)$ or $(0,1)$. However, because $X$ is compact, it is impossible. Hence, we have that $n\ge 2$. By Steps 7-8, there exists a $C^k$ diffeomorphism $\psi_2:V_2\to I_2$ such that $I_2$ is $S^1$, $[0,1]$, $(0,1)$, or $[0,1)$. Because $X$ is compact, $V_2=X$ implies that $I_2$ is either $S^1$ or $[0,1]$, and by Step 6, the converse is also true. Therefore, $I_2$ is $S^1$ or $[0,1]$ if and only if $n=2$, and if $n>2$, then $I_2$ is either $(0,1)$ or $[0,1)$. Moreover, if $I_2=[0,1)$, then $\psi_2^{-1}(0)\in \partial X$. Hence, $(\varphi_3,\psi_2)$ satisfies condition B. Inductively, we can assume that for each $j\in \{2,...,n\}$, there is a $C^k$ diffeomorphism $\psi_j:V_j\to I_j$ that satisfies the same condition as $\psi_2$, and because $V_n=X$, $I_n$ is either $S^1$ or $[0,1]$. This completes the proof of Step 9.
\end{proof}

Step 9 says that our initial claim is correct. This completes the proof.
\end{proof}

\section*{Acknowledgments}
I am grateful to an anonymous referee for his/her helpful suggestions. This work was supported by JSPS KAKENHI Grant Number JP21K01403.

\section*{Reference}

\begin{description}

\item{[1]} Arrow, K. J., Block, H. D., and Hurwicz, L. (1959) ``On the Stability of the Competitive Equilibrium, II.'' Econometrica 27, pp.82-109.

\item{[2]} Arrow, K. J. and Debreu, G. (1954) ``Existence of an Equilibrium for a Competitive Economy''. Econometrica 22, pp.265-290.

\item{[3]} Debreu, G. (1952) ``Definite and Semi-Definite Quadratic Forms.'' Econometrica 20, pp.295-300.

\item{[4]} Debreu, G. (1970) ``Economies with a Finite Set of Equilibria.'' Econometrica 38, pp.387-392.

\item{[5]} Evans, L. C. (2010) {\it Partial Differential Equation: Second Edition}. American Mathematical Society, Providence.

\item{[6]} Gale, D. (1987) ``The Classification of 1-Manifolds: A Take-Home Exam.'' American Mathematical Monthly 94, pp.170-175.

\item{[7]} Guillemin, V. and Pollack, A. (1974) \textit{Differential Topology}. Prentice Hall, New Jersey.

\item{[8]} Hirsch, M. W. (1976) {\it Differential Topology}. Springer, New York.

\item{[9]} Mas-Colell, A. (1985) \textit{The Theory of General Economic Equilibrium: A Differentiable Approach}. Cambridge University Press, New York.

\item{[10]} Milnor, J. W. (1965) {\it Topology from the Differential Viewpoint}. The University Press of Virginia, Charlottesville.

\end{description}

\end{document}